\documentclass[]{article}
\usepackage[utf8]{inputenc}
\usepackage{authblk}
\usepackage[margin=1in]{geometry}

\usepackage{mathpazo}
\usepackage{amssymb}
\usepackage{amsmath}
\usepackage{mathtools}
\usepackage{relsize}
\usepackage{hyperref}
\usepackage{amsthm}
\usepackage{physics}

\usepackage{subcaption}

\usepackage{comment}
\usepackage[sort&compress,numbers]{natbib}
\renewcommand{\Pr}{\mathbb{P}}

\renewcommand{\vec}{\operatorname{vec}}

\newcommand{\htsp}{\mspace{0.5mu}}

\newcommand{\I}{\mathbb{1}}

\renewcommand{\ip}[2]{\langle #1 , #2\rangle}

\renewcommand{\ket}[1]{
  \lvert\htsp #1 \htsp \rangle}

\renewcommand{\bra}[1]{
  \langle\htsp #1 \htsp \rvert}

\newcommand{\cA}{\mathcal{A}}
\newcommand{\cB}{\mathcal{B}}

\newcommand{\cF}{\mathcal{F}}

\newcommand{\cM}{\mathcal{M}}
\newcommand{\cN}{\mathcal{N}}

\newcommand{\cP}{\mathcal{P}}
\newcommand{\cQ}{\mathcal{Q}}
\newcommand{\cR}{\mathcal{R}}

\newcommand{\cW}{\mathcal{W}}
\newcommand{\cX}{\mathcal{X}}
\newcommand{\cY}{\mathcal{Y}}

\newcommand{\Pos}{\mathrm{Pos}}
\newcommand{\Herm}{\mathrm{Herm}}

\newcommand{\Density}{\mathrm{D}}

\newcommand{\ol}{\overline}
\newcommand{\ve}{\varepsilon}

\newcommand{\bbC}{\mathbb{C}}

\theoremstyle{definition}
\newtheorem{theorem}{Theorem}
\newtheorem{lemma}[theorem]{Lemma}

\newtheorem{definition}{Definition}

\newtheorem{proposition}{Proposition}

\title{Limits of Detecting Extraterrestrial Civilizations}
\author{Ian George, Xinan Chen, Lav R. Varshney}
\affil{Department of Electrical and Computer Engineering, and Coordinated Science Laboratory, University of Illinois Urbana-Champaign, Urbana, IL  61801}

\begin{document}

\maketitle
    
\begin{abstract}
    The search for extraterrestrial intelligence (SETI) is a scientific endeavor which struggles with unique issues--- a strong indeterminacy in what data to look for and when to do so. This has led to attempts at finding both fundamental limits of the communication between extraterrestrial intelligence and human civilizations, as well as benchmarks so as to predict what kinds of signals we might most expect. Previous work has been formulated in terms of the information-theoretic task of communication, but we instead argue it should be viewed as a detection problem, specifically one-shot (asymmetric) hypothesis testing. With this new interpretation, we develop fundamental limits as well as provide simple examples of how to use this framework to analyze and benchmark different possible signals from extraterrestrial civilizations. We show that electromagnetic signaling for detection requires much less power than for communication, that detection as a function of power can be non-linear, and that much of the analysis in this framework may be addressed using computationally efficient optimization problems, thereby demonstrating tools for further inquiry.
\end{abstract}

SETI, the search for extraterrestrial intelligence (ETI), has now been considered in a serious manner for more than half a century. The signals to be searched for, ranging from electromagnetic pulses to strange observations of exoplanets, perhaps even artificial ones \cite{Arnold05a, Arnold05b, Arnold13}, lie in a continually growing set \cite{reyes2019comprehensive,sheikh_2020,wright2021strategies,houston2021strategies}. The scientist must therefore not only propose possible signals that may be detected, but also prioritize which signals are more likely. One insight about SETI is that it may be framed as a communication task, which allows information-theoretic analysis of different possible signals in terms of resources and likelihood of success \cite{Shannon48}. Indeed, information-theoretic formulations of SETI have provided results such as: the energy per bit of information is lower for inscribed matter than for electromagnetic signals \cite{Rose04,hippke18inscribedmatter}, the limitations on the number of bits that can be sent by various media \cite{shostak11limits}, a Bayesian framework for the probability that Milky Way-based ETI-generated electromagnetic signals cross Earth's orbit \cite{grimaldi2018bayesian}, benchmarking of different information carriers for signals to/from extraterrestrial civilizations \cite{messerschmitt15,hippke2018benchmarking,hippke18inscribedmatter}, analyses of communication networks at an interstellar scale \cite{hippke2020interstellar,messerschmitt12}, and even whether the cooling of the universe could give a thermodynamic advantage to reversible computation such that aliens are simply waiting to communicate in the far future \cite{aestivation16,bennett19}.

We follow previous works in presenting both fundamental limits of SETI as well as providing new quantitative tools for benchmarking different possible ETI signals. Yet, we do not view SETI as a communication problem in which large amounts of information must be transmitted. Instead, we view SETI as a detection problem--- a yes or no answer to the question `was an extraterrestrially-engineered process detected from a given spacetime region?’ This strengthens a previous insight that an `attention-getting signal' may only require a few bits \cite{Arnold13}. Moreover, as many possible ETI signals may be only observed once rather repetitively, and we do not know when we will receive the message to begin with, we need a one-shot analysis of detection. For these reasons we view SETI as the information-theoretic task of one-shot hypothesis testing.

Using this interpretation, we present the general optimization problem for constructing the optimal decision function for given ETI processes (processes engineered by ETI) that minimizes false negatives while guaranteeing a chosen threshold of false positives. Conveniently, this problem is computationally efficient. Moreover, this characterization allows us to determine the fundamental measure for hypothesis testing. We then use this new mathematical framework to look at previously proposed ETI signals. We also include examples to illustrate conceptual points about the formulation as well as the numerical implementation of the framework. We show that the power demands of electromagnetic signals have been overestimated previously and that the advantage of inscribed matter is less obvious for detection.

\section{SETI as One-Shot Hypothesis Testing}\label{sec:SETIasOneShot}
As mentioned, \textit{communicating} and \textit{detecting} are, in general, related but distinct tasks. Moreover, detecting whom one wishes to communicate with is necessary before trying to extract information from a physical system that is encoding a message. Abstractly, one can view the receiver in SETI as taking in massive amounts of data in various forms and trying to determine if it is the result of natural processes (the null hypothesis) or processes engineered by an extraterrestrial intelligence (the alternative hypothesis). Without essential loss of generality, we can think of the receiver device as taking in many bands of electromagnetic waves.\footnote{We say without essential loss of generality in the sense that we presume sound and smell will not be used, and even if the signal were not electromagnetic itself, such as for inscribed matter, one must interact with it through electromagnetic radiation (sight, sensing equipment, etc.) which may then be data-processed to more coarse-grained properties.} Formally, by discretizing time and power of the input signals, at each time step the detector reads an outcome $y \in \cY$, where $\cY$ is a finite alphabet of possible detector readings. Therefore, if the time monitored includes $n$ time bins, the total observed sequence would be $y^{n} \in \cY^{\times n}$.

As such, SETI supposes the sequence came from either the natural process distribution,  $P_{0}$, or the ETI process distribution, $P_{1}$. Note that $P_{1}$ may be a mixture of many possible extraterrestrial received signals (detected sequences). The receiver then must `decode' whether this sequence corresponds to $P_{0}$ or $P_{1}$. As noted in \cite[Section 11.7]{cover2006}, an intuitive way of `decoding' is to define a set of sequences $Y_{0} \subset \cY^{\times n}$ so the null hypothesis is declared if a sequence in $Y_0$ is observed and the alternate hypothesis is declared otherwise. We however relax this decision function to be probabilistic, and sufficiently general to encapsulate both quantum mechanical and classical signals. (See SI for a brief introduction to quantum information theory.) Given $\cY^{\times n}$, we define a vector space $\cY^{\otimes n}$ with basis vectors pertaining to the possible sequences. Probability distributions over the original sequence space $\cY^{\times n}$, such as the null and alternative hypotheses, $P_0, P_1$, can then be written as diagonal matrices contained in the space of linear operators from $\cY^{\otimes n}$ to itself. We will denote the space of probability distributions over $\cY^{\times n}$ in the vector space by $\mathcal{P}(\cY^{\otimes n})$. A deterministic classical decision function can then be written as a projector $\Pi_{Y_0} = \sum_{y \in Y_{0}} \dyad{y}$ where we have used bra-ket notation. The type I and type II error probabilities of this decision function, $\alpha_{n}$ and $\beta_{n}$ respectively, can then be expressed as
\begin{align}\label{eqn:errorProbs}
    \alpha_{n} = \Tr(P_{0}\Pi_{Y_{0}}^{\perp})  \quad   \beta_{n} = \Tr(P_{1}\Pi_{Y_{0}}) \ ,
\end{align}
where $\Pi_{Y_{0}}^{\perp}$ is the projector onto the orthogonal complement of $\Pi_{Y_{0}}$.
Given this intuition, we then can generalize to consider any (quantum) probabilistic decision function using the transformations:
$$ P_{0} \to \rho \hspace{2mm} P_{1} \to \sigma \hspace{2mm} \Pi_{Y_{0}} \to  0 \leq A_{n} \leq \mathbb{I}_{\cY^{\otimes n}} \hspace{2mm} \Pi_{Y_{0}}^{\perp} \to \mathbb{I}_{\cY^{\otimes n}} - A_{n} \ ,$$
where $\mathbb{I}_{\cY^{\otimes n}}$ is the identity matrix on $\cY^{\otimes n}$, $\rho,\sigma$ are (finite) quantum probability distributions over $\cY^{\otimes n}$, and $A_{n}$ is an arbitrary positive semidefinite operator such that $\mathbb{I}_{\cY_n}-A_n$ is also positive semidefinite. It follows $\{A_n,\,\mathbb{I}_{\cY_n}-A_n\}$ forms a positive-operator-value-measurement (POVM), which means it represents a physically implementable measurement device. In the sequel, we remain agnostic to whether $P_{0},P_{1}, A_{n}$ are quantum or classical unless stated explicitly. We denote the set of quantum probability distributions as $\mathrm{D}(\cY^{\otimes n})$. Note that $\mathcal{P}(\cY^{\otimes n}) \subset \mathrm{D}(\cY^{\otimes n})$, and so any statement that holds in the quantum case includes the more common classical case.

Before we move forward, we note two points. First, the arrival time of a signal in SETI is unknown. As such, its starting time in the sequence is best viewed as distributed over all time steps according to a (hidden) random variable $\nu$. Furthermore, for any length of time $n \in \mathbb{N}$ that we consider, $n$ should be seen as greater than the value that upper bounds the arrival time $\nu$. This is because if the signal never arrived, one would be considering the signal sampled from $P_{0}$ regardless. Also note that if there were a fixed delay, $\nu$ would no longer be a random variable, but a constant, reducing $P_{1}$ to a distribution which samples from $P_{0}$ until some time $t$. This means the test should be unaffected by the delay since one can remove the elements of each sequence for times less than $t$. This mirrors the insight that communication capacity with fixed delay is the same as capacity with no delay \cite{Chandar12}.

\subsection{Information-Theoretic Limits of Detection of ETI}

With this formalization, we can view the SETI problem from the sender's side or from the receiver's side. We begin with the receiver's problem of constructing the optimal binary test. This depends on the electromagnetic radiation one is trying to detect, as this determines the alternative hypothesis $P_{1}$ and the cone of space the detector takes information from, and in turn determines the noisy channel and the null hypothesis $P_{0}$. With these fixed, the problem is just designing the optimal decision function. This is fundamentally a one-shot problem where at best the receiver saves the whole sequence and adaptively updates the decision function as more signals are acquired (i.e. as $n$ increases). Assuming $P_{0}$ and $P_{1}$ are not mutually exclusive/orthogonal, the error probabilities \eqref{eqn:errorProbs} cannot be zero, and so the tradeoff between these must be considered. 

Given the consequence of claiming the detection of ETI, the probability of false alarm $\alpha_{n}$ should be bounded above by $\varepsilon$, which is small, and then the probability of missing an ETI's signal $\beta_{n}$ should be minimized under this constraint \cite{Neyman1933}. Formally, the receiver aims to minimize $\beta_{n}$ while guaranteeing $1-\alpha_{n} \geq 1 - \varepsilon$ for $\varepsilon \in [0,1)$. To write the optimization problem cleanly, note that
\begin{align}\label{eqn:alphaIdentity}
\alpha_{n} = \Tr(P_{0}(\mathbb{I}_{\cY^{n}} - A_{n})) = 1 - \Tr(P_{0}A_{n}) \ , 
\end{align}
where the second equality follows from the unit trace of $P_{0}$, as it is a (possibly quantum) probability distribution. By the definition of $\beta_{n}$, \eqref{eqn:errorProbs}, and the identity \eqref{eqn:alphaIdentity}, the optimization problem that determines the optimal $\beta_{n}$ in this setting is:
\begin{align}\label{eqn:minBeta}
    \xi^{\ve}_{H}(P_{0}||P_{1}) := \underset{\substack{0 \leq A_{n} \leq \mathbb{I}_{\cY^{n}} \\ \Tr(P_{0}A_{n}) \geq 1 - \varepsilon}}{\min} \Tr(P_{1} A_{n}) \ .
\end{align} 
Note that if $P_{0}, P_{1}$ are classical, the optimal decision function $A_{n}$ is also classical. Moreover, in this case, the optimizer is guaranteed to be diagonal, so \eqref{eqn:minBeta} reduces to a linear program (LP) which can be efficiently solved for large data sets. More generally, \eqref{eqn:minBeta} is always a semidefinite program (SDP), so if $\cY^{n}$ is small, it can be efficiently evaluated. Most interestingly, \eqref{eqn:minBeta} is the argument of the $\varepsilon$-hypothesis testing relative divergence $D^{\varepsilon}_{H}(P_{0}||P_{1}) := -\log(\xi^{\ve}_{H}(P_{0}||P_{1}))$ \cite{wang12}, which is therefore the fundamental limit for ETI detection.

The generalized Quantum Stein's Lemma \cite{generalizedQuantumStein} determines the fundamental limit of one-shot hypothesis testing for a large class of approximately repetitive signals. Informally,\footnote{See the Supplementary Information for the formal statement.} it states that if for any number of time steps $n \in \mathbb{N}$, the set of possible alternative hypotheses, $\cM_{n}$, is closed, convex, only contains permutation-invariant hypotheses, and satisfies a few other consistency conditions, then there exists a sequence of decision functions $\{A_{n}\}_{n \in \mathbb{N}}$ such that the asymptotic type II error $\beta_{n}$ is given by:
\begin{align}\label{eqn:genQuantumStein}
\lim_{n \to \infty} \frac{1}{n} D^{\ve}_{H}(P_{0}^{\otimes n}||P_{1,n}) = \lim_{n \to \infty} \frac{1}{n} \underset{P_{1,n} \in \cM_{n}}{\min} D(P_{0}||P_{1,n}) \ .
\end{align}
The requirement that the null hypothesis is identically and independently distributed (i.i.d.\ ) is not restrictive since natural processes, such as noise from space, are i.i.d.\ The primary limitation of this theorem is that it requires the alternative hypotheses to be permutation-invariant over time. However, for repetitive signals, this seems reasonable as one can often model the initial signal as i.i.d.\ over some time scale, and as long as the noise the signal experiences is memoryless, the received signal will be i.i.d. 

In the case the ETI signals that are finite time, we prove convergence of the decision function (see SI for derivation).
\begin{theorem}\label{thm:finite-length-convergence}
For all $\varepsilon \in [0,1]$, a finite set $\mathcal{Q}$ of signals with finite maximum length and a finite set of possible i.i.d.\ null hypotheses has an optimal asymptotic type II error achieved in finite time. Moreover, it may be determined using a semidefinite program.
\end{theorem}
\noindent Note this holds for multiple i.i.d.\ null hypotheses unlike the generalized quantum Stein's lemma. This result along with the quantum Stein's lemma covers all cases relevant to SETI except unbounded uncertain arrival time and infinite length non-periodic signals, neither of which one would expect to converge in general. We refer to the Supplementary Information for further analysis.

\paragraph*{Quantum or Classical Information}
Given these limits for the sender, we might ask if quantum signals could provide any advantage. We first note that one common issue in using quantum information to one's advantage is the need for aligned \textit{reference frames} \cite{bartlett07}. This arises because one uses quantized degrees of freedom of an object to transmit data (for example, the polarization of single photons). It follows one not only needs a \textit{codebook} for quantum information theory, but a notion of \textit{alignment} between the two parties' reference frames (for example, party $A$ may send a single photon in what they view as horizontal polarization, but at arrival it is diagonal polarization as defined by party $B$). In this work, this issue is not explicit because by defining the distributions to evaluate \eqref{eqn:minBeta}, one would have to define the hypotheses in the fixed local reference frame. We do note however this suggests whatever signal is to be sent should not rely on the alignment of reference frames, which at least complicates the advantages of quantum mechanics. The possibility of macrosopic quantum signals \cite{Friedman00} or superpositions of degrees of freedom that do not rely on a reference frame \cite{Loredo19} do, however, allow the possibility of quantum ETI signals.

If quantum signals are possible, it is further possible they provide an advantage. For example, one possibility for an advantage comes from sending highly entangled states. If the receiver assumes detecting many entangled states is not a common natural process, then under the i.i.d.\ assumption and in the asymptotic limit, the type II error rate would be non-zero for any classical null hypotheses, and the worst case asymptotic behavior would be given by the relative entropy of entanglement \cite{vedral1998entanglement} of the alternative hypothesis. This is distinct from the classical case, since if the alternative hypotheses are classical, no such general claim could be made. Yet, our current technologies would not be able to preserve entanglement between particles sent over many light years. Beyond this aspect, we note that for all relevant classical distributions, if quantum and classical signals were equally achievable and the assumed local reference frame were correct or not relevant, we can say that quantum signals could only help, as is stated in the following proposition. The proposition can be viewed as an immediate consequence of the data-processing inequality for $\xi^{\ve}_{H}$ (proven in the SI).
\begin{proposition}\label{prop:QuantumCanOnlyHelp}
For any classical distributions $P_0,P_1 \in \mathcal{P}(\cY^{\otimes n})$, if quantum signals are implementable, we can achieve at least the same optimal type II error, $\beta_{n}$, using quantum signals. Moreover, there exist cases where the advantage is strict.
\end{proposition}

\paragraph{Signal Design} Moving from receiver-side to sender-side analysis, the goal is to construct a signal that can be decoded with small error probability. Given the previous analysis, there are two options. The first is to simply construct a `single' signal which could not be generated by nature (with almost any probability). This leads to $P_0$ and $P_1$ being largely orthogonal, which allows for a hypothesis test such that $\alpha_{n} < \varepsilon$ and $\beta$ being small as follows from optimization problem \eqref{eqn:minBeta}. Such signals, however, may require significant energy. Examples of such `single' signals include inscribed matter \cite{Rose04} and possibly overhead meteors, which we discuss in subsections \ref{subsec:InscribedMatter} and \ref{subsec:OverheadMeteors} respectively. The second is to send a series of i.i.d.\ signals through a memoryless channel to utilize the Quantum Stein Lemma. To do so, the senders ought to conjecture a memoryless channel $\mathcal{N}$ so as to build a memoryless device $D$ such that the received message is distinguishable from the sender's assumed noise model, at least asymptotically. Formally, if the device and noise are memoryless processes, the hypotheses are of the form $P_{0} = \cN(\star)^{\otimes{n}}$ and $P_{1} = \cN(P_{D})^{\otimes n}$ where $P_{D}$ is the distribution over signals the device produces each pulse and $\star$ is the distribution when the device is off, which may be taken to be vacuum. Given these conditions, by the Quantum Stein Lemma, the transmitter aims to construct a device, i.e. distribution, such that its source distribution satisfies some set of constraints. Denoting the feasible set of devices under said constraints as $\mathcal{C}$, the optimal source is then
\begin{align}\label{eqn:SenderFundamental}
\underset{P_{D} \in \mathcal{C}}{\sup} D(\cN(\star)||\cN(P_{D})) ,
\end{align}
where we note that this optimization problem can be unbounded if there exists $P_{D} \in \mathcal{C}$ such that the null hypothesis does not lie in the support of the alternative hypothesis. In that special case, there is perfect distinguishability asymptotically. Examples of signals for which this optimization problem applies would be radio signals or the use of transits, which we discuss in subsections \ref{subsec:RadioSignals} and \ref{subsec:Transits} respectively.

\subsection{Analyzing Measured Data}\label{subsec:AnalyzingObtainedData}
Note that our discussion so far has been in terms of comparing processes as a way of evaluating preferable methods of sending/receiving ETI announcements while taking the one-shot nature of the problem seriously. Moreover, the decision function once computed could be used on incoming data, though this would be under the assumption the incoming data was truly from one of the two hypotheses. One problem that may seem somewhat distinct from this is that of having obtained data and then trying to make a decision based on this data. This does not make much of a difference in the one-shot setting as there is no difference (assuming classical data) between storing the data and then implementing the decision function \textit{once} and implementing the decision function as the data comes in:\footnote{This is only true if one does not condition on something in the observed data.} 
\begin{enumerate}
    \item Denote the obtained data by $d \in \cY^{n}$.
    \item Construct model(s) $\{M_{i}\}_{i}$ which give rise to probability distribution(s) $\{Q_{i}\}_{i}$ over $\cY^{n}$ such that the probability of $d$ is non-zero for each.
    \item Choose $\ve \in (0,1)$. For each $Q_{i}$, solve for the optimal $\beta_{n}$ error \eqref{eqn:minBeta}. While the value of $\beta_{n}$ is not relevant in the case of obtained data, the optimizer of the problem is the optimal decision function $A_{i}^{\ve}$, where we have added the superscript $\ve$ as it is also a function of $\ve$ in general. 
    \item Let $\ket{d}$ represent the obtained data $d$ in the vector space. 
    \item  If $\ket{d} \in \text{Ker}(A_{i}^{\ve})$, then the data should certainly not be considered evidence of the alternative hypothesis. 
    \item Otherwise, implement the decision function and apply it to input $d$.
\end{enumerate}
Although a positive decision is not definitive in determining the presence of ETI, as one cannot guarantee the assumed model holds, if nothing else this gives a rigorous way of eliminating possible data for any given model when $\ket{d} \in \text{Ker}(A^{\ve}_{i})$ by using our framework. Of course, for finite sets of finite length signals (Theorem \ref{thm:finite-length-convergence}) or convex sets of i.i.d.\ signals \eqref{eqn:genQuantumStein}, strong conclusions may be drawn for reasonable sets of hypotheses. In other words, the basic approach can be extended to any generalized hypothesis testing setting such as universal or composite testing \cite{Levitan02,Berta21}.

\section{Analyzing Specific Kinds of Signals}
Having introduced a new formalism for analyzing detection and announcement for SETI, we now consider previously-proposed signaling methods for announcing the existence of a civilization under this framework both in terms of fundamental limits as well as using numerical tools availed to us by $\ve$-hypothesis testing being an SDP (LP for classical distributions). We consider both orthodox and unorthodox proposals to better see the generality of framing SETI as hypothesis testing. 

\subsection{Electromagnetic Signals} \label{subsec:RadioSignals}
Perhaps the most orthodox approach to SETI is sending radio signals, though more recently the consideration of laser signals (continuous wave and laser pulse) has grown \cite{kingsley01}. Roughly speaking, in this approach the limiting factor is the power of the transmitter \cite{Shannon49}. Clearly if the signal had enough power, the signal would be detectable, much in the same way the capacity of a Gaussian channel is limited by the power. In principle there is the further issue of how many planets the civilization would like to signal at once which will increase the number of transmitters necessary (and the total amount of necessary energy). If one could generate a sufficiently powerful burst from a laser, assuming it were detected, it would be sufficient. However, it is commonly held that an ETI would more likely periodically pulse a laser at their target, due to the limitation of generating sufficient energy. For periodic signals, the longer the time-span the signal is sent, the closer one is to achieving the Stein's lemma limit in our framework. Therefore we can make predictions about an optimal transmitter under given power constraints using \eqref{eqn:SenderFundamental}. For example, in the case of an average and peak power constraint, \eqref{eqn:SenderFundamental} might be written as:
    \begin{align}\label{eqn:EMSignalFundamental}
    \underset{\substack{\mathbb{E}_{P}(P_{init}) \leq \ol{P}_{av}
    \\ f(P_{init}) \leq P_{\max}}}{\max}
    D(\mathcal{N}_{r} \circ \mathcal{N}_{t} \circ \mathcal{N}_{s}(\star^{n})||\mathcal{N}_{r} \circ \mathcal{N}_{t} \circ \mathcal{N}_{s}(P_{init})) \ ,
\end{align}
where $P_{init}$ is the initial distribution of the signal, $f$ is a function that calculates the power cost, $\overline{P}_{av}$ is the upperbound on the power, $\mathcal{N}_{r}$ is the noise at the receiver's end, $\mathcal{N}_{t}$ is the noise during the transmission, $\mathcal{N}_{s}$ is the noise from the sender's end, and $\mathbb{E}_{P}$ is the expectation of $f(P_{init})$. Fixing the noise models, this gives one a close approximation of the fundamental limit of the distinguishability of the signals and the probability of false positive detection as a function of $P_{av}$ using \eqref{eqn:SenderFundamental}.

\subsubsection{Distinguishability Does Not Universally Necessitate Strong Signal}

We now consider a simple but counterintuitive example using this equation. A more in-depth derivation is presented in the Supplementary Information. As the example is fully classical, for simplicity we view probability distributions as vectors in bra-ket notation, so that we can express the distributions by the non-zero probability sequences. 

Consider a pulsed laser. Assume one discretizes the total signal as a sequence of length $n \in \mathbb{N}$. For clarity, we let $n = 5$. The alphabet for each element of the sequence is the interval $[0,g] \subset \mathbb{N}$ by discretizing the power and choosing a cutoff for the possible power of an observed signal.\footnote{One reason for such a cutoff is tolerated input of the device.} Assuming the laser is a square pulse, the expected optimal choice of the initial distribution could be written as $\ket{0,P,0,0,0,}$, i.e. a delta distribution. We can imagine that while there is no noise at the source, there is memoryless jitter in the laser which with probability $q/2$ shifts the sequence forward or backward one time bin. We therefore define the distribution
$$P_{init} = (1-q) \ket{0,0,P,0,0} + \frac{q}{2} \left( \ket{0,P,0,0,0} + \ket{0,0,0,P,0} \right) \ . $$
We assume that the noise during travel $\mathcal{N}_{t}$ is loss-only, so for each time bin, the map $y \to \max(y-c,0)$ is applied, where $c$ is a function of the distance travelled and possibly the conditions over the travel path. Finally, we assume the noise at the receiver is the composition of two maps. First we assume the data is taken over a short enough time (as lasers can pulse reasonably quickly) that the sun is additive power so that for each time bin the map $y \to \min(y+s,g)$. The second map assumes with some probability there is any given possible sequence.\footnote{Technically, the introduction of this map is to guarantee absolute continuity for the sake of our example. The ad-hoc introduction of such a map to guarantee this is largely an aspect of the simplicity of our model. However it is in general a rigorous way to guarantee both the null and alternative hypothesis are full rank so as to guarantee a finite value, and, by the data-processing inequality for relative entropy along with the Chernoff-Stein lemma \cite[Theorem 11.8.3]{cover2006}, we know we can only have made the asymptotic error exponent worse by doing this.} This is modeled by a linear map on distributions, $\mathcal{N}_{r,2}: q \mapsto (1-\delta)q + \frac{\delta}{|\cY^{\times n}|} \vec{1}_{\cY^{\times n}}$, where $\vec{1}$ is the all-ones vector and $\delta \in (0,1)$. Given these maps, one can determine $P_0,P_1$ from $P_{init}$. Under the assumption $c<P<g-s+c$, one finds that so long as $c \neq P$, $D(P_0||P_{1})$ is the same constant. As the assumption implies $c < P$ implies $c \neq p$, the asymptotic error rate for all powers in this range is the same. 

While this model is extremely simple, and so we would not expect such independence to hold in standard cases, it exemplifies the important conceptual aspect of the problem: the goal of the signal is to distort the sample enough to be distinguishable from the pure noise case, and it is not a priori necessary that the transmission must significantly overpower the noise to achieve this. We note this point had been made previously quantitatively in \cite{kingsley93}, this example simply shows a particularly simple situation where this holds. 

Finally, while in many cases the optimal signalling device may seem obvious (pick the largest average power allowed), in more elaborate cases it may not be the case, which would give an advantage to analyzing the process using \eqref{eqn:EMSignalFundamental}. Moreover, even if the optimal is straightforward, we will see in a later example (Subsection \ref{subsec:InscribedMatter}) that in the one-shot setting the error probability of the optimal decision function may not scale linearly in resources for generating the ETI process, in which case further tradeoffs may be worth considering. 

\subsection{Near-Earth Projectiles} \label{subsec:OverheadMeteors}
A less orthodox approach of recent interest is near-earth projectiles. Most generally, we take near-earth projectile signals as the construction of any series of macroscopic objects (projectiles) which are directed in a trajectory which passes near the Earth without colliding into it. It seems unlikely that it would be most efficient to construct a large number of such objects, and it seems perhaps most rational to expect there to be only a few such projectiles in the message. In this case the one-shot nature of the problem is very important as the data will only be collected once before the projectiles continue on their overhead trajectory or burn up in the atmosphere, and so it is crucial to have some notion of error probability of a false positive for such a signal detection, which is exactly what the one-shot hypothesis testing interpretation provides.

Indeed, the ability to handle the false positive probability in this setting has become a reality given the recent interstellar object that passed through our solar system, `Oumuamua, and the debate as to whether its origins were ETI or natural \cite{Loeb21,bannister19}. While there have been arguments that `Oumuamua was of ETI origin \cite{Loeb21}, analysis concluded `Oumuamua was most likely of natural origins, while noting for some not-yet-explained aspects \cite{bannister19}. The arguments presented in \cite{bannister19} are largely about considering different observed properties of `Oumuamua and how they deviate from expected observations. This is exactly what one-shot hypothesis testing does in a mathematical sense. Note that this is trying to make a conclusion on already obtained data, and so the methodology of Subsection \ref{subsec:AnalyzingObtainedData} applies.

\subsubsection{Simple Numerical Example}
We consider a toy example of overhead meteors to show the application of our framework. Meteors often burn up in our atmosphere. Indeed, this happens consistently enough that it is used as a tool in telecommunications by bouncing signals off of the ionization trail of the meteors, known as meteor bursts \cite{meteorBurstCommunication}. Both meteors simply falling and meteor bursts are generally held to be Poisson processes, and both have data consistently collected on them, so a reasonable approach to an ETI signal would be to produce meteors that differ from how we expect, as this would at least be recorded. Therefore, assuming an ETI that knows that meteors are not uniformly distributed on every planet, an ETI may send a small number of meteors in rapid uniform succession at the Earth for a short period. It would make sense for the meteors to be small to save energy and to guarantee they do not harm the Earth. We therefore can compare the meteor detection when the ETI meteors are and are not included and look at how distinguishable the two cases are.

Mathematically, as meteor bursts are a Poisson process, for any interval $\Delta t$, $\Pr[n \text{ meteors}] = \frac{(\lambda \Delta t)^{n}}{n!}e ^{\lambda \Delta t}$, where $\lambda$ may depend on many things, such as the time of day and of year \cite{meteorBurstCommunication}. For simplicity, we assume a scaling such that $\Delta t = 1$ and assume that at least the start of the ETI message begins in this time interval. We look at the probability of missing an ETI signal $\beta_{n}$ as a function of how many ETI rocks appear over the time interval for two choices of $\lambda$ and three choices of $\varepsilon$ which corresponds to maximum allowed error probability $\alpha_{n}$ using \eqref{eqn:minBeta} to compare the original Poisson process and the Poisson process with this additive noise. We note that our numerical analysis must be finite, whereas the Poisson process has a countably infinite number of outcomes. This can be rigorously handled by truncating the tail of the distribution, given the tail property of the distribution and the data-processing inequality. For completeness, this is elaborated on in the Supplementary Information. We numerically construct results for our simple example in Figure \ref{fig:Simple_Example}.

\begin{figure}[ht]
\begin{subfigure}{\columnwidth}
  \centering
  \includegraphics{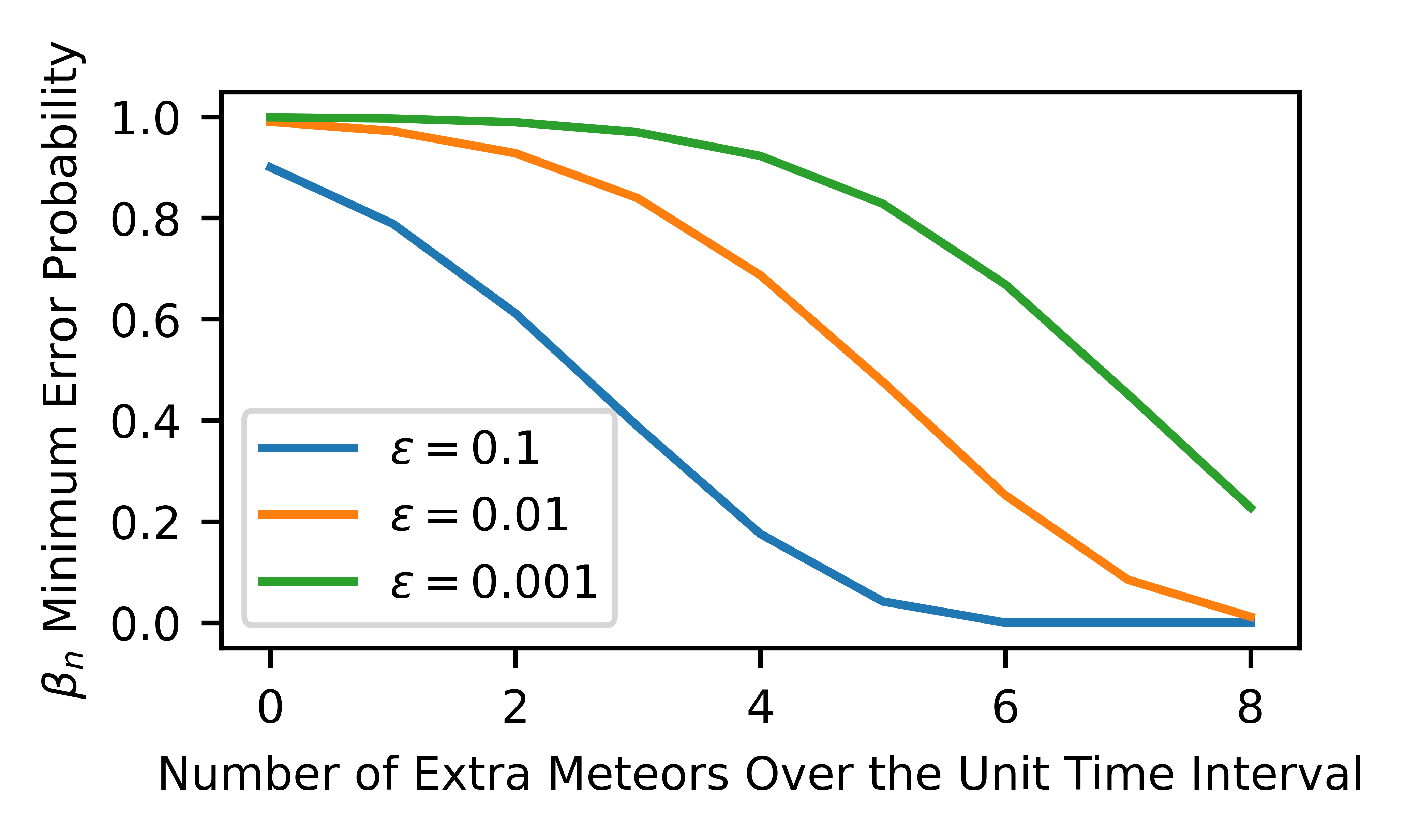}
  \caption{}
  \label{fig:sub-first}
\end{subfigure}
\begin{subfigure}{\columnwidth}
  \centering
  \includegraphics{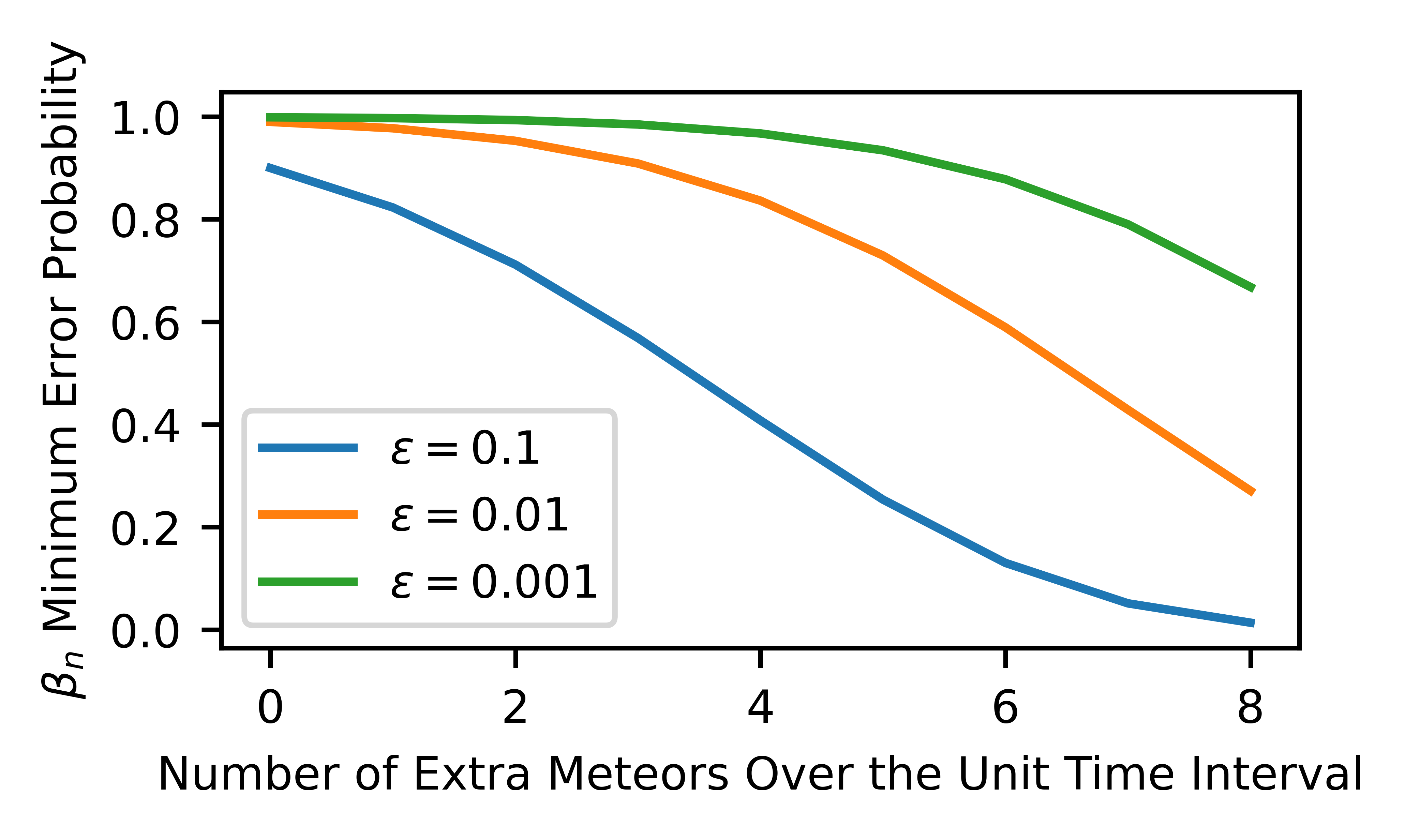}  
  \caption{}
  \label{fig:sub-second}
\end{subfigure}
\caption{Here we see the error probability of the optimal decision function for determining if there were extra detected meteors in a unit interval as a function of the extra meteors. \textbf{(a)} is for natural meteors distributed according to $\mathrm{Pois}(3)$, and \textbf{(b)} is for $\mathrm{Pois}(6)$ (i.e. larger variance). Each plot is plotted for three different tolerances of false positive probability $\alpha_{n} < \ve$.}
\label{fig:Simple_Example}
\end{figure}

General properties of the method can be seen in Figure \ref{fig:Simple_Example}. First, the optimal error probability $\beta_{n}$ does not in general scale linearly in $\alpha_{n}$. This can be seen in the graphs as $\alpha_{n} \geq 1 - \ve$ is the constraint and as $\varepsilon$ shrinks, $\beta_{n}$ does not change linearly. This is one reason it is advantageous for it to be computationally efficient to construct the optimal decision function for more general models. Second, it is highly sensitive to the null hypothesis. For example, we see that as the Poisson distribution broadens (as $\lambda$ goes from three to six) and because the ETI signal is additive, the number of meteors to signal with low risk of false negatives increases. This implies an increase in needed energy for this method, thereby allowing one to infer the feasibility/cost-benefit of this signal from the ETI's perspective conditioned on their knowledge of the receiver's local atmospheric conditions.

\subsection{Inscribed Matter} \label{subsec:InscribedMatter}
A related but distinct approach to near-earth projectiles is inscribed matter \cite{Rose04}. In \cite{Rose04}, the authors show that under many circumstances one can encode and transmit more information for less energy by encoding information densely into matter. It follows that the known advantage of inscribed matter only holds for messages carrying a lot of information, which is the opposite information-theoretic regime than the one-shot detection problem we consider here. Regardless, one can view communicating with inscribed matter as a signal for ETI detection that simply happens to have more data on-board--- a strategy that seems quite reasonable and has been studied in other settings \cite{Varshney13,Varshney19}. However, if one sends inscribed matter (that should get trapped in orbit about some planet or successfully crash land), one would expect the energy cost to largely be a function of the distinguishability from its new local environment, as was observed in our previous example. In particular, since $\beta_{n}$ increases under data processing, one expects one wants macroscopic design properties that would not get coarse-grained so as to be detectable, and this could lead to energy costs not considered in \cite{Rose04}.\footnote{In \cite{hippke18inscribedmatter}, they suggest that the optimal method is to send inscribed matter shielded in a long cylinder. It is not clear that this would be optimal for detection however.} This question, at least from the sender side, is well suited for investigation via a (non-linear) variation of the one-shot hypothesis testing optimization problem:
\begin{equation}\label{eqn:power_constrained_inscribed_matter}
\begin{aligned}
    & \underset{0 \preceq A \preceq I, P_{d} \in \mathrm{D}(\cX)}{\text{minimize}} && \ \ \langle P_{1} , A \rangle \\
    & && \langle P_0, A \rangle \geq 1 - \ve \\
    & && P_{1} = \mathcal{N}_{n}(P_{d}) \\
    & && E(P_{d}) \leq B
\end{aligned}
\end{equation}
where, $\cX$ is the Euclidean space the designed signal $P_{d}$ is defined over, $E(\cdot)$ is an energy cost function, $B$ is a constraint on the total energy, and $\mathcal{N}_{n}$ is a linear map representing the noise introduced to the design during transmission. Therefore, we believe the one-shot hypothesis testing framework remains relevant for inscribed matter approach, because, while it makes a good case for it being more energy efficient, it does not escape the detection problem that one-shot hypothesis testing encompasses.

\subsection{Transit} \label{subsec:Transits}
One final unorthodox method proposed for SETI is to look for signals of large extraterrestrially-engineered objects orbiting stars, referred to as artificial transits \cite{Arnold05a,Arnold05b,Arnold13}. The initial motivation for this approach is that we might achieve such detections in our search for exoplanets because the stellar flux detected is dependent on the shape of the transit, and so an artificial transit with strange shape could be detected. This was numerically demonstrated in \cite{Arnold05a} where the author compared various simulations of transit signals. 

Like the previously mentioned methods, the primary limitation seems to be energy. Whereas the advantage of inscribed matter was the amount of data that can be sent as a function of energy, the advantage of artificial transits is both that they could be discovered in standard astronomical research and that they can stay in orbit for a long time. The advantage of this presented in \cite{Arnold13} is that it allows for a signal (the pulsed stellar flux) over a much longer time scale than pulsed electromagnetic radiation from a laser, which is both limited by continuously generating power and the decline of a given civilization. Artificial transits also have the advantage of not needing to be aimed like the other methods generally would, as they continue to orbit around the star, i.e.\ they seem to be the best proposed broadcast signal to-date.

The duration of the orbiting process suggests another advantage of constructing transits beyond those given in \cite{Arnold05a,Arnold05b,Arnold13}, which is that it is the longest lasting i.i.d.\ signal and so the most promising to achieve the fundamental limit of hypothesis testing.\footnote{There is the previously noted caveat that the signal will only be approximated as i.i.d.\ if the noise is memory-less, but this problem is not unique to artificial transits.} This only adds to the credibility of this possible method. 

Finally, we note an open problem in the quantitative analysis of SETI through one-shot hypothesis testing that may be beneficial for future analysis. In \cite{Arnold05a} it is noted that artificial transits whose projected cross-section is triangular produce a detected stellar flux waveform that is similar to the waveform generated by a planet with rings which could complicate at least that choice of cross-section's ability to be detectable. Using the one-shot hypothesis testing optimization program on discretized waveforms could provide a stronger understanding of this particular complication. Of course, this could then be extended to analyze various signal forms, providing a quantitative benchmark for which signals forms seem more likely.

\section{Conclusion \& Outlook}
In this work we have presented a new interpretation of SETI as one-shot hypothesis testing. The crux of the argument is that communication with ETI civilizations is a related but distinct task to that of detecting said civilizations. Specifically, detection is the answer to a `yes' or `no' question, `is this process natural or generated by extraterrestrial intelligence?,' whereas communication is exchanging significantly more information. Moreover, we stress the one-shot aspect of SETI as the signals we are trying to detect may be brief and/or non-i.i.d. Using these insights we show how SETI can be formalized as one-shot hypothesis testing, and present the optimization problem which constructs the optimal decision function for a hypothesis test between an ETI process and a natural process, where optimality is in terms of minimum false negatives given some demand on the rarity of false positives for the decision function. In the special case of (mixtures of) periodic ETI signals, such as from a pulsed laser or an artificial transit, we use the generalized Quantum Stein's lemma \cite{generalizedQuantumStein} to recover the fundamental limit of hypothesis testing in these settings. This in turn dictates how the sender should design their signal, \eqref{eqn:SenderFundamental}.

To clarify that viewing SETI as a case of one-shot hypothesis testing does not hold only in the abstract, we considered various proposals for ETI signals and explained how they relate to the one-shot hypothesis testing framework. Moreover, we presented a simple numerical example to illustrate how to analyze given proposals by making use of the data-processing property of hypothesis testing and that the one-shot hypothesis testing optimization problem \eqref{eqn:minBeta} is a semidefinite program in the general case and a linear program in the common case where all data is classical.

As final remarks, we note where one could further this line of investigation. The most natural critique with this formulation is that we are considering binary hypothesis testing, and so one cannot make conclusions about multiple alternative hypotheses at the same time. Of course binary hypothesis testing stemmed from the argument that we are trying to detect a `yes' or `no,' question--- we are not concerned with comparing different alternative hypotheses, just distinguishing whether a received signal is from an ETI civilization or not. This conceptual point is lost if we consider discriminating between $M$ alternative hypotheses. Secondly, we have shown the ease of the one-shot hypothesis testing optimization problem \eqref{eqn:minBeta}. Indeed, in the case where all the data is classical, the problem becomes a linear program and, given the modern state of linear programming, this could be implemented over realistic large data sets in reasonable time for practical benefit. However, the optimization problems pertaining to cost constraints \eqref{eqn:SenderFundamental}, \eqref{eqn:power_constrained_inscribed_matter} are not so trivial as they are concave and nonlinear optimization problems respectively.\footnote{\eqref{eqn:SenderFundamental} is only a concave optimization program if the constraint set $\mathcal{C}$ is convex. Otherwise it is a nonlinear optimization problem as well.} Due to this, while we think the problems could be useful regardless and possibly even tractable some of the time, further investigation would be necessary for their application. Regardless, we believe viewing SETI as one-shot hypothesis testing can provide rigorous quantitative analysis techniques and lead to new insights in the field.

\section*{Acknowledgments}
We acknowledge Kartik Kumar Kansal as a member of the initial project that led to this work.

\section{Supplementary Information}
\subsection{Quantum Information Background and Previous Results}
In this appendix we state the quantum information definitions, constructions, and properties that we will need. We refer to \cite{WatrousBook,Wilde2011,Tomamichel2016} for further information.

\begin{definition}
Given a finite alphabet $\Sigma$, there exists a complex Euclidean space $\cX = \mathbb{C}^{|\Sigma|}$. Denote the set of positive semidefinite operators over this space as $\mathrm{Pos}(\cX)$. A quantum probability distribution over $\Sigma$ is $\rho \in \mathrm{Pos}(\cX)$ such that $\Tr(\rho) = 1$. The set of quantum probability distributions over the space $\cX$ is denoted $\mathrm{D}(\cX)$. We denote the set of classical probability distributions over the same basis by $\mathcal{P}(\cX)$.
\end{definition}
A quantum probability distribution is a generalization of a probability distribution as by the spectral decomposition theorem $\rho$ may always be diagonalized, and, as it is trace one and positive semidefinite, this diagonalization has all entries between zero and one such that they sum to one. In other words, the diagonalization of a quantum probability distribution is a classical probability distribution over some preferred basis. Furthermore, sampling i.i.d. from a quantum state $\rho \in \mathrm{D}(X)$ $n$ times is the same as imagining you have access to $n$ copies of the state, i.e. considering the state $\rho^{\otimes n}$ is the quantum version of considering $X_{i} \sim P$ $n$ times for some distribution $P$ over $\cX$, where $\otimes$ is the Kronecker product. 

\begin{definition}[Bra-Ket Notation]
For a fixed basis of a complex Euclidean space of dimension $n$, the basis vectors can be denoted by $\{\ket{0},\ket{1},\hdots,\ket{n}\}$. Where $\ket{i}$ is the column vector with a one in the $i^{th}$ element. 
\end{definition}

In the main text we make use of taking classical Cartesian products of sequences of elements of finite alphabets, $y^{n}  \in \cY^{\times n}$ to a vector space $\cY^{\otimes n}$. This can be done in the following manner as is clear from the above definition. Imagine that $y^{n} \in \cY^{\times n}$ is a $n$-length sequence of elements from a finite alphabet $\cY = \{a_{1},a_{2},\hdots, a_{k} \}$. By overloading notation, define a complex Euclidean space $\cY \cong \mathbb{C}^{n}$ whose basis vectors are $\{\ket{a_{i}}\}_{i=1}^{k}$. Then using the Kronecker product, $\otimes$, the complex Euclidean space $\cY^{\otimes n}$ is a vector space whose basis vectors are the $n$-length sequences of the basis vectors $\{\ket{a_{i}}\}$. That is to say, $\cY_{\otimes n}$ is a vector space whose basis vectors are the possible sequences contained in $\cY_{1}^{\times n}$.

\begin{definition}[Born's Rule]
Given a complex Euclidean space $\cX$ and the space of linear operators over $\cX$ to itself, $L(\cX)$, Born's rule states that for an operator $0 \leq A \leq \mathbb{I}_{\cX}$ and a quantum probability distribution $\rho \in \mathrm{D}(\cX)$, the probability of observing the property $A$ corresponds to is given by $\Tr(\rho A)$, $p(A) = \Tr(\rho A)$.
\end{definition}

\begin{definition}
Given finite alphabets $\cX, \cY$, a quantum measurement on $\cX$ with output over $\cY$ is a function $\mu: \cY \to \text{Pos}(X)$ that satisfies the constraint $$ \sum_{y \in \cY} \mu(y) = \mathbb{I}_{\cX} \ .$$
\end{definition}
This can be interpreted as saying the outcome of the measurement device are indexed by $\cY$, and the measurements that result in said outcome are defined by the operators $\mu(y)$ by Born's rule. Note that a decoder that outputs classical symbols over a finite alphabet $\cY$ from an input linear space over $\cX$ is always a quantum measurement $\mu: \cY \to \text{Pos}(X)$.

\subsection*{One-Shot Hypothesis Testing Background Results}
Here we state the results relevant to the one-shot hypothesis testing: the generalized Quantum Stein's Lemma \cite{generalizedQuantumStein}, its reduction to the standard Quantum Stein's Lemma, and the Data Processing Inequality for one-shot hypothesis testing. 
\begin{theorem}
\cite{generalizedQuantumStein}
Let $\{\cM_{n}\}_{n \in \mathbb{N}}$ be a set of sets satisfy the following for all $n \in \mathbb{N}$: 
\begin{enumerate}
    \item $\cM_{n} \subset \mathrm{D}(\cX^{\otimes n})$ is closed and convex
    \item $\cM_{n}$ contains $\sigma^{\otimes n}$ where $\sigma$ is full rank
    \item If $\rho \in \cM_{n+1}$, $\Tr_{k}(\rho) \in \cM_{n}$ for all $k \in [n+1]$ 
    \item If $\rho \in \cM_{n}, \nu \in \cM_{m}$, $\rho \otimes \nu \in \cM_{m + n}$
    \item If $\rho \in \cM_{n}$, $W_{\pi}\rho W_{\pi}^{\ast} \in \cM_{n}$ for all $\pi \in \mathcal{S}_{n}$ where $W_{\pi}$ is the unitary that permutes the $n$ copies of $\cX$ according to the permutation $\pi$.
\end{enumerate}
Under these conditions, it holds that for any $\rho \in \mathrm{D}(\cX)$, there exists a sequence of two-outomce POVMS $\{A_{n},I-A_{n}\}$, such that $\alpha_{n} \to 0$ and  for all $n \in \mathbb{N}$, $\omega_{n} \in \cM_{n}$,
\begin{align}
\lim_{n \to \infty} \frac{1}{n} D^{\ve}_{H}(\rho^{\otimes n}||\omega_{n}) = \lim_{n \to \infty} \frac{1}{n} \underset{\omega_{n} \in \cM_{n}}{\min} D(\rho^{\otimes n}||\omega_{n}) \ .
\end{align}
\end{theorem}
\noindent Note this can be simplified to the standard (Quantum) Stein Lemma by letting $\cM_{n} := \{\sigma^{\otimes n}\}$ for all $n \in \mathbb{N}$ as then the equation simplifies to $\lim_{n \to \infty} \frac{1}{n}D^{\ve}_{H}(\rho^{\otimes n}||\sigma^{\otimes n}) = D(\rho||\sigma)$, which implies a single equation to evaluate for the i.i.d. case. 

The requirement that $\{\cM_{n}\}_{n \in \mathbb{N}}$ are closed convex sets allows this fundamental limit to hold when one wishes to consider mixtures of multiple i.i.d. alternative hypotheses at the same time. This cannot be handled by the standard Quantum Stein Lemma as convex combinations of i.i.d.\ signals are not i.i.d.\ in general. As mentioned in the main text, for non-repetitive signals, where permutation invariance over time does not hold for many signals, the limit of $\xi^{\ve}_{H}$ as time $n$ grows must be considered directly. We note that if one models the unknown arrival time of the signal by a random variable $\nu$, the alternative hypothesis will not be permutation invariant over time. This is because the alternative hypothesis will be of the form
$$P_{1} = \Tr_{\mathcal{V}}(P_{1}^{<\nu} \otimes S \otimes P_{1}^{>\Delta t + \nu} \otimes \dyad{\nu})$$ where $P_{1}^{<\nu}$ is the distribution before the signal arrives at time $\nu$, $S$ is the distribution of the signal, $P_{1}^{>\Delta t + \nu}$ is the distribution after the signal ends at time $\Delta t + \nu$, and $\Tr_{\mathcal{V}}$ is the partial trace over the random variable $\nu$, constructing the hypothesis from which we sample. We note that there are asymptotic limit results that hold for the classical regime for handling Gauss-Markov processes \cite{Sung06} and hypothesis testing for sets of null and alternative hypotheses in the ergodic settting \cite{Luschgy93}. Moreover, in the quantum setting, if one can guarantee they are only interested in i.i.d.\ signals, \cite{Berta21} may be used instead of the generalized Quantum Stein's lemma.

In Section \ref{subsec:finite_set_hyp_testing} below, we prove for a finite set of i.i.d.\ null hypotheses and a finite set of alternative hypotheses which end in finite time and return to the i.i.d.\ null hypothesis converge in finite time (Theorem \ref{thm:finite_convergence}). This applies for any finite length message as well as unknown arrival time distributions, so long as the unknown arrival time is assumed to take a maximum value $t_{a}$ such that with the longest finite time message $t_{m}$, $t_{a} + t_{m} \leq t_{end}$, where $t_{end}$ is the final time bin for the decision function. As the generalized Quantum Stein's lemma implies convergence for i.i.d. signals, the only condition not handled by these two results is when the arrival time and message length do not satisfy the above conditions. However, in that case it is not obvious this would in general converge.

Finally, we present the following lemma which allows us to conclude that the error exponent can only increase under not only data-processing, but also post-selection of the data. We note this property has been proven in the literature previously \cite{Wang19}.
\begin{lemma}[Data-Processing of $\varepsilon$-Hypothesis Testing for Trace Non-Increasing Maps]\label{lemma:DPTNI}
Let $\mathcal{N}: \cX \to \cY$ be a completely-positive trace non-increasing map (i.e. if $P \succeq 0$, $\cN(P) \succeq 0$ and for all operators $X$, $\Tr(X) \geq \Tr(\cN(X))$). Let $P_0 \in \mathrm{D}(\cX), P_1 \succeq 0$, and $\ve \in (0,1)$. It follows, $$\xi^{\ve}_{H}(P_{0}||P_{1}) \leq \xi^{\ve}_{H}(\cN(P_{0})||\cN(P_{1})) \quad \text{and} \quad  D^{\ve}_{H}(P_{0}||P_{1}) \geq D^{\ve}_{H}(\cN(P_{0})||\cN(P_{1})) \ . $$
\end{lemma}
\begin{proof}
We prove the property for the $\xi_{H}^{\ve}$ inequality, the other follows directly from the definition of $D^{\ve}_{H}$ in terms of $\xi_{H}^{\ve}$. Let $A$ be feasible for $\xi^{\ve}_{H}(\cN(P_0)||\cN(P_{1}))$, i.e. $\Tr[A \cN(P_0)] \geq 1 - \ve$. By definition of the adjoint of a linear map, we have
\begin{equation}\label{eq:adjointfeas1}
\Tr[\cN^{\dagger}(A)P_0] = \Tr[A \cN(P_0)] \geq 1 - \ve \ . 
\end{equation}
Moreover, $\cN^\dagger$ is completely-positive and sub-unital ($\cN^{\dagger} (\mathbb{I}) \leq \mathbb{I}$) as $\cN$ is completely-positive and trace non-increasing. As $0 \preceq A \preceq \mathbb{I}$, 
\begin{equation}\label{eq:adjointfeas2}
    0 \preceq \cN^{\dagger}(A) \preceq \cN^{\dagger}(\mathbb{I}) \preceq \mathbb{I} \ .
\end{equation}
The relations \eqref{eq:adjointfeas1} and \eqref{eq:adjointfeas2} imply that $\cN^{\dagger}(A)$ is a feasible point for $\xi^{\ve}_{H}(P_0||P_1)$. Therefore, using that $\xi^{\ve}_{H}$ is a minimization problem and that we considered arbitrary freasible $A$ (so we have included the minimizer) we may conclude $\xi^{\ve}_{H}(\cN(P_0)||\cN(P_1)) \geq \xi^{\ve}_{H}(P_0||P_1)$. 
\end{proof}

\subsubsection*{Proof of Proposition 1}
\begin{proof}[Proof of Proposition 1]
The first half of the proof is an immediate consequence of Lemma \ref{lemma:DPTNI} as can be seen in the following manner. Consider any classical distributions $P_{0},P_{1} \in \mathcal{P}(\cY^{n}).$ Recall these are diagonal matrices. Consider any quantum states $\rho, \sigma \in \mathrm{D}(\cY^{n})$ such that the diagonal entries of $\rho,\sigma$ are the same as $P_0,P_1$ respectively. Consider the pinching channel $\mathbb{P}$ which zeros out every non-diagonal entry of a matrix. It is easy to verify this is trace-preserving and completely positive, and thus a channel. Thus by data-processing (Lemma \ref{lemma:DPTNI}), for any $\ve \in (0,1]$,
Consider any 
$$ \xi^{\ve}_{H}(\rho||\sigma) \leq \xi^{\ve}_{H}(\mathbb{P}(\rho)||\mathbb{P}(\sigma)) = \xi^{\ve}_{H}(P_{0}||P_{1}) \ . $$
This proves we can always do at least just as well with quantum signals.

To show there exist cases where the quantum advantage is strict, we consider the following simple but physically reasonable example. Let $P_0 = \begin{pmatrix} 1 & 0 \\ 0 & 0 \end{pmatrix}$ and $P_1 = \frac{1}{2}\begin{pmatrix} 1 & 0 \\ 0 & 1 \end{pmatrix}$. This models the scenario where the null hypothesis is signal 1, whereas the alternative hypothesis contains either signal 1 or signal 2 with equal probabilities. Now, for any $\ve>0$, the requirement that $\langle P_0,A\rangle \geq 1-\ve$ implies $\langle P_1,A\rangle \geq \frac{1}{2}\langle P_0,A\rangle \geq \frac{1}{2}(1-\ve)$, where the first inequality follows from the fact that $A$ is positive. Therefore $\beta \geq \frac{1}{2}(1-\ve)$. However, if the alternative hypothesis is $P_1 = \frac{1}{2}\begin{pmatrix} 1 & 1 \\ 1 & 1 \end{pmatrix}$ (i.e., a quantum superposition of signal 1 and 2), then we can choose $A' = \begin{pmatrix} 1-\ve & -\sqrt{\ve(1-\ve)} \\ -\sqrt{\ve(1-\ve)} & \ve \end{pmatrix}$ (note that $0\preceq A'\preceq \mathbb{I}$). This particular detection gives $\langle P_0, A'\rangle = \Tr(P_0 A') = 1-\ve$, making it a feasible decision function, and $\langle P_1, A'\rangle =\Tr (P_1 A')= \frac{1}{2}\left(1-2\sqrt{\ve(1-\ve)}\right) < \frac{1}{2}(1-\ve)$ if $\ve \in (0,\frac{4}{5})$. This proves that there exist quantum strategies that achieves error probability strictly less than $\beta$. 
\end{proof}

\subsection{Decision Function Convergence in Finite Time for Finite Sets of Finite Time Messages}\label{subsec:finite_set_hyp_testing}
In this section we prove that if all signals are finite time and the distribution returns to the i.i.d.\ null hypothesis after the signal ends, then the decision function can be guaranteed to converge. The primary lemma is to show when the alternative and null hypothesis become the same and independent of the message, the decision function will not improve. While intuitive, to the best of our knowledge, this has not been proven in the generality we consider. Our primary tool in proving this is strong duality for semidefinite programs which we summarize first. This is the presentation given in \cite{WatrousBook}, except that we have the primal problem be a minimization rather than dual. This is known to be equivalent.

\subsubsection*{Semidefinite Programs and Duality Theory}
Let $\Psi \in \mathrm{T}(\cA,\cB)$ be a Hermitian-preserving map, $A \in \Herm(\cX)$, and $B \in \Herm(\cY)$. A semidefinite program is a triple $(\Psi,A,B)$, with the following associated optimization problems: \\

\begin{minipage}{.5\linewidth}
\begin{equation}\label{eqn:primaldef}
        \begin{aligned}
                & {\text{minimize}} & & \langle A, X \rangle \\
                & \text{subject to} & & \Psi(X) = B \\
                &                   & &  X \in \Pos(\mathcal{X})
        \end{aligned}
\end{equation}
\end{minipage}%
\begin{minipage}{.5\linewidth}
\begin{equation}\label{eqn:dualdef}
        \begin{aligned}
                & {\text{maximize}} & & \langle B, Y \rangle \\
                & \text{subject to} & & \Psi^{\ast}(Y) \preceq A  \\
                &                   & &  Y \in \Herm(\mathcal{Y})
        \end{aligned}
\end{equation}
\end{minipage}
where $\Psi^{\ast}$ is the adjoint map of $\Psi$. \eqref{eqn:primaldef} is referred to as the \textit{primal problem} and \eqref{eqn:dualdef} is referred to as the \textit{dual problem}. We define 
$$\cF_{\cA} = \{X \in \Pos(\mathcal{X}) | \Psi(X) = B\} \text{ and } \cF_{\cB} = \{Y \in \Herm(\mathcal{Y}) | \Psi^{\ast}(Y) \preceq A\} \ . $$ These sets are referred to as the \textit{feasible set} of the primal problem and dual problem, respectively. 

 By \textit{weak duality}, for all semidefinite programs, the optimal value of the primal problem, denoted by $\alpha$, is always greater than or equal to the optimal value to the dual problem, denoted by $\beta$. If a semidefinite program has that $\alpha = \beta$, it is said to have \textit{strong duality}. A sufficient condition to show strong duality for SDP is Slater's condition.
\begin{theorem}{\textit{(Slater's Condition)}} \label{thm:SlatersCondition}
For a semidefinite program $(\Psi,A,B)$,  if $\alpha$ is finite and there exists a Hermitian operator $Y$ which \textit{strictly} satisfies the dual problem, that is, $\Psi^{\ast}(Y) \prec A$, then $\alpha = \beta$ and the optimal value is obtained in the primal problem.
\end{theorem}

\subsubsection{Derivation of Result}
First note that we could re-express $\xi^{\ve}_{H}(\rho_{0}||\rho_{1})$ in the following manner:  
\begin{align}\label{eqn:hyp_test_gamma}
 \min \{ \gamma : 
        \ip{X}{\rho_{0}} \geq 1 - \ve, \,
        \ip{X}{\rho_{1}} \leq \gamma, \,
        X \preceq \I, \,
        \gamma \geq 0, \, 
        X \in \Pos(\cX) 
        \} \ .
\end{align}
This obtains the same value as as for any feasible $X$ the optimal choice for $\gamma$ is $\gamma := \ip{X}{\rho_{1}}$. It is then straightforward to generalize this to consider a finite number of null and alternative hypotheses. Let $\cP, \cQ \subset \Pos(\cX)$ be finite sets with index alphabets $\Lambda, \Sigma$ respectively. Then we define the following optimization problem:
\begin{equation}\label{eqn:multi-alt-hyp-primal}
    \begin{aligned}
        \xi^{\ve}_{H}(\cP||\cQ) := \text{minimize:}\quad & \gamma & \\
        \text{subject to:}\quad & \ip{X}{P^{i}_{0}} \geq 1 - \ve & \forall i \in \Lambda &\\
        & \ip{X}{P^{j}_{1}} \leq \gamma & \forall j \in \Sigma \\
        & X \preceq \I & \\
        & \gamma \geq 0 , \, X \in \Pos(\cX) & \ ,
    \end{aligned}
\end{equation}
where it is clear that if $\cP = \{\rho_0\}, \cQ = \{\rho_{1}\}$, the problem simplifies to $\xi^{\ve}_{H}(\rho_0 || \rho_1)$, which justifies the notation. Furthermore, if $\cP, \cQ \subset \mathrm{D}(\cX)$, then $X$ is the optimal decision function such that the type I error $\alpha_{n}$ is less than $\ve$ for all null hypotheses and $\gamma$ is the minimum type II error $\beta_{n}$ that will hold for all the alternative hypotheses at once. This confirms this is the definition we want.

As we will be using duality theory, we next derive the dual problem. To do that, we convert \eqref{eqn:multi-alt-hyp-primal} into the standard form of \eqref{eqn:primaldef}. Doing this, we obtain
\begin{align*}
    A &= (1 \oplus 0_{\cX} \oplus \left( \oplus_{i=1}^{|\Lambda|} 0 \right) \oplus \left(\oplus_{j=1}^{|\Sigma|} 0 \right) \oplus 0_{\cX} \oplus 0) \in \Pos(\bbC \oplus \cX \oplus \bbC^{\oplus |\Lambda|} \oplus \bbC^{\oplus |\Sigma|} \oplus \cX \oplus \bbC) \\
    B &= \left(\oplus_{i=1}^{|\Lambda|}(1-\ve) \oplus \left(\oplus_{j=1}^{|\Sigma|} 0 \right) \oplus \I \oplus 0 \right) \in \Pos(\bbC^{\oplus |\Lambda|} \oplus \bbC^{\oplus |\Sigma|} \oplus \cX \oplus \bbC) \\
    X &= \widetilde{\text{diag}}(\gamma, X, \oplus_{i=1}^{|\Lambda|} c_{i}, \oplus_{j=1}^{|\Sigma|} d_{j}, M, f) \in \Pos(\bbC \oplus \cX \oplus \bbC^{\oplus |\Lambda|} \oplus \bbC^{\oplus |\Sigma|} \oplus \cX \oplus \bbC) \\
    \Psi(X) &:= \text{diag}\left( \left(\bigoplus_{i=1}^{|\Lambda|} \Phi_{P^{i}_{0}}(X) - c_{i} \right), \left(\bigoplus_{j=1}^{|\Sigma|}\Phi_{P^{j}_{1}}(X) - \gamma + d_{j} \right), X + M, \gamma - f \right) \ ,
\end{align*}
where $c_{i},d_{j},M,f$ are all slack variables to satisfy the equality in the standard form \eqref{eqn:primaldef}, $\Phi_{Z}(X) := \ip{Z^{\ast}}{X} = \Tr(ZX)$, $\widetilde{\text{diag}}$ means that the operator is defined over the whole space but for our purposes we only need to label these diagonal blocks, and $\text{diag}$ means that the operator is block-diagonal. 

With this, we need to obtain the adjoint map of $\Psi$. To do this, we first note that the adjoint map of $\Phi_{Z}(\cdot)$ is $\cdot \otimes Z$ for Hermitian $Z$ as can be verified:
$$ \langle \Tr(ZX), c \rangle = \langle \sum_{i} \bra{i} ZX \ket{i}, c \rangle = \langle ZX , c\I \rangle = \langle X, cZ \rangle \ , $$
where we used that we assumed $Z$ is Hermitian. Using this, one can determine that the adjoint map is given by
$$ \Psi^{\ast}(Y) = \text{diag}(w-\sum_{j=1}^{|\Sigma|}v_{j}, \hspace{1mm} \sum_{i=1}^{|\Lambda|} z_{i}P^{i}_{0} + \sum_{j=1}^{|\Sigma|} v_{j} P^{j}_{1} + Z,\hspace{1mm}
\oplus_{i=1}^{|\Lambda|} -z_{i},\hspace{1mm}
\oplus_{j=1}^{|\Sigma|} v_{j},\hspace{1mm}
Z,\hspace{1mm}
-w ) \ . $$
where we have used $Y = \widetilde{\text{diag}}(\oplus_{i=1}^{|\Lambda} z_{i}, \oplus_{j=1}^{|\Sigma|}v_{j}, Z, w)$. This allows one to write down the dual problem which after simplification is of the form 
\begin{equation}\label{eqn:multi_alt_hyp_dual_simp}
        \begin{aligned}
                & {\text{maximize}} & & (1-\ve)\|\mathbf{z}\|_{1} - \Tr(Z) \\
                & \text{subject to} & & \sum_{i=1}^{|\Lambda|} z_{i} P^{i}_0 - \sum_{j=1}^{|\Sigma|} v_{j}P^{j}_1 \preceq Z \\
                &                   & & \sum_{j = 1}^{|\Sigma|} v_{j}\leq 1, \, 0 \leq v_{j} \, \forall j \\
                &                   & & 0 \preceq \mathbf{z},Z \ .
        \end{aligned}
\end{equation}
With the primal and dual problem specified, we can now prove strong duality of the problem.
\begin{lemma}\label{lemma:strong_duality}
    \eqref{eqn:multi-alt-hyp-primal} and \eqref{eqn:multi_alt_hyp_dual_simp} satisfy strong duality.
\end{lemma}
\begin{proof}
    We will use Slater's condition (Theorem \ref{thm:SlatersCondition}). First note $\alpha$ is always finite as it is lower bounded by zero. Therefore all we need is to prove there exists a dual feasible point with strict feasibility. Let $\delta \in (0,\frac{1}{|\Sigma|})$. Let $z_{i} = \delta$ for all $i$ and $v_{j} = \delta$ for all $j$. Let $Z = 2\delta (\sum_{i} P^i_0 - \sum_j P^j_1)^{+} + \delta \I$, where $(H)^{+}$ is the positive eigenspace of the Hermitian operator $H$. Then all the inequalities hold strictly, and so this is a solution which is strictly feasible.
\end{proof}

\begin{lemma}\label{lemma:tensor_prod_factorization}
    Let $\cP, \cQ \subset \Pos(\cX)$ be finite sets and $\ve \in [0,1]$. It holds $$\xi^{\ve}_{H}(\cP || \cQ ) = \xi^{\ve}_{H}(\cP \otimes \omega || \cQ \otimes \omega) \ , $$
    where $\cR \otimes \omega := \{R_{k} \otimes \omega : R_{k} \in \cR \}$ and $\omega \in \Density(\cW)$.
\end{lemma}
\begin{proof}
Let $(\gamma^{\star},X^{\star})$ and $(\mathbf{z}^\star, \mathbf{v}^\star, Z^\star)$ be the optimizers for the primal and dual problem of $\xi^{\ve}_{H}(\cP||\cQ)$ respectively. They obtain the same value as strong duality holds under these settings (Lemma \ref{lemma:strong_duality}). That is to say, $\gamma^\star = (1-\ve)\|\mathbf{z}^\star\| - \Tr(Z^\star)$. Consider $X = X^\star \otimes \I_{\cW}$ and $\gamma = \gamma^\star$. This is feasible for the primal problem of $\xi^{\ve}_{H}(\cP \otimes \omega || \cQ \otimes \omega)$ as
$$ \ip{X \otimes \I_{\cW}}{P_{0}^{i} \otimes \omega} = \ip{X}{P_{0}^{i}} \geq 1 - \ve \ ,$$
where we have used the multiplicativity of trace over tensor products and that $\omega$ has unit trace. Next consider $Z = Z^\star \otimes \omega$, $z_{i} = z_{i}^\star$, $v_{j} = v_{j}^\star$. This is feasible as
$$ \sum_{i}^{|\Lambda|} z^{\star}_{i}P_{0}^{i} \otimes \omega - \sum_{j}^{|\Sigma|} v^{\star}_{j}P^{j}_{1} \otimes \omega  \preceq Z^{\star} \otimes \omega \Leftrightarrow \left(\sum_{i=1}^{|\Lambda|} z^{\star}_{i} P^{i}_0 - \sum_{j=1}^{|\Sigma|} v^{\star}_{j}P^{j}_1 \right) \otimes \omega \preceq Z^{\star} \otimes \omega \ ,$$
which is always true given $(\mathbf{z}^{\star},\mathbf{v}^{\star},Z^{\star})$ being the optimizer of $\xi^{\ve}_{H}(\cP||\cQ)$. Furthermore, $(1-\ve)\|\mathbf{z}^{\star}\| - \Tr(Z^{\star} \otimes \omega) = (1-\ve)\|\mathbf{z}^{\star}\| - \Tr(Z^{\star})$ by again using the unit trace of $\omega$. Thus we have constructed primal and dual optimizers that achieve the same value, showing they are optimal, and obtain the same optimal value as $\xi^{\ve}_{H}(\cP||\cQ)$, which completes the proof.
\end{proof}

With this property of the one-shot hypothesis testing SDP proven, we may prove that the optimal decision function converges in finite time for finite size messages if the null hypothesis is i.i.d.\ and the received signal returns to the null hypothesis after the end of the transmitted message. We note this final point is not limiting as if it were not to return to the null hypothesis, then in effect the message has not ended.

\begin{theorem}\label{thm:finite_convergence}
For all $\varepsilon \in [0,1]$, given a finite set $\cQ$ of finite length signals of maximal alphabet size $\cX^{\otimes k}$ and a finite set of possible i.i.d. null hypotheses, $\cP := \{\rho^{\otimes n} : \rho \in \eta \}$, the optimal asymptotic type II error is achieved in finite time and characterized by $\xi^{\ve}_{H}(\cP||\cQ)$.
\end{theorem}
\begin{proof}
Let $\varepsilon \in [0,1]$. We begin with a single null hypothesis. Consider a null hypothesis $\rho^{\otimes n} \in \Density(\cX^{\otimes n})$ for any $n \in \mathbb{N}$. Consider a finite set of possible signals, $\cQ \subset \Density(\cX^{\otimes k})$. Under the assumption that once the signal ends one starts sampling from the null hypothesis, the type II error is given by $\beta_{n}~=~\xi^{\ve}_{H}(\rho^{\otimes n}||\cQ \otimes \rho^{\otimes (n-k)})$ for all $n \geq k$. By lemma \ref{lemma:tensor_prod_factorization}, $\xi^{\ve}_{H}(\rho^{\otimes n}||\cQ \otimes \rho^{\otimes (n-k)}) = \xi^{\ve}_{H}(\rho^{\otimes k}||\cQ)$ for all $n \geq k$. Thus in this setting $\lim_{n \to \infty} \beta_{n} = \xi^{\ve}_{H}(\rho^{\otimes k}||\cQ)$.

Now consider a finite set $\eta \subset \Density(\cX)$. We assume that once the message ends, it samples from the considered null hypothesis. Formally, this means the only null and alternative hypothesis pairs assumed acceptable are of the form $\left(\rho^{\otimes n}, \sigma \otimes \rho^{\otimes (n-k)}\right)$ where $\rho \in \eta, \sigma \in \cQ$, and we stress that $\rho$ is the same in both parts of the pair. It follows for fixed $\rho \in \eta$, the relevant quantity is $\xi^{\ve}(\rho^{\otimes k}||\cQ)$ as this is the single null hypothesis case. This holds for all $\rho \in \eta$. It follows that the minimal asymptotic type II error which has at most $\ve$ type I error for $\cP := \{\rho^{\otimes k} : \rho \in \eta \}$ is given by $\xi^{\ve}_{H}(\cP||\cQ)$. That is, in this setting, $\lim_{n \to \infty} \beta_{n} = \xi^{\ve}_{H}(\cP||\cQ)$.
\end{proof}
    
\subsection{Models for Examples}
In this section we provide complete calculations of the examples in the text.

\subsubsection*{Distinguishability Does Not Universally Necessitate Strong Signal}
We consider a finite alphabet $\cY := [0,g] \subset \mathbb{N}$ and sequence $y^{n} \in \cY^{\times n}$, where $n \in \mathbb{N}$. The distribution over this sample space will be sampled an infinite number of times in an iid manner. As noted in the main text, the whole discussion will hold for sequences of any length $n \in \mathbb{N}$, but we let $n =5$ for clarity. Following the exposition in the main text, we consider a signal generated by a laser which sends a pulse in one time bin of every five, but suffers from time jitter so that the pulse may be shifted by one time bin forward or backward with probability $q/2$ for both. This leads to an initial probability distribution
$$P_{init} = (1-q) \ket{0,0,P,0,0} + \frac{q}{2} \left( \ket{0,P,0,0,0} + \ket{0,0,0,P,0} \right) \ . $$
We assume that the noise during travel $\mathcal{N}_{t}$ is loss-only and acts on every time bin identically by removing an amount of power $c$, which is a function of the distance travelled and possibly the conditions over the travel path. The map may be defined on its action on the basis sequences as 
$$\mathcal{N}_{t}:  \ket{y_{1},y_{2},\hdots,y_{5}} \mapsto \ket{\max(y_{1} - c,0),\max(y_{2} - c,0), \hdots, \max(y_{5}-c,0)} \ , $$ for all sequences $(y_1, \hdots, y_5) \in \cY^{\times 5}$. We assume the noise at the receiver is the composition of two maps. First we assume the data is taken over a short enough time that the sun is additive power and can be described as the linear map 
$$\mathcal{N}_{r,1}:  \ket{y_{1},y_{2},\hdots,y_{5}} \mapsto  \ket{\min(y_{1} + s,g), \min(y_{2} + s,g), \hdots , \min(y_{5} + s,g)}$$ for all sequences $(y_1, \hdots , y_5) \in \cY^{\times 5}$. The second map assumes with some probability there is any given possible sequence. \footnote{Technically, we assume the existence of this map so as to guarantee the support of the alternative hypothesis is contained in the support of the null hypothesis for the sake of our example. The ad-hoc introduction of such a map to guarantee this is largely an aspect of the simplicity of our model. However it is in general a rigorous way to guarantee both the null and alternative hypothesis are full rank so as to guarantee a finite value, and, by the data-processing inequality for relative entropy along with the Chernoff-Stein lemma \cite[Theorem 11.8.3]{cover2006}, we know we can only have made the asymptotic error exponent worse by doing this.} This may be defined as $$\mathcal{N}_{r,2}: q \mapsto (1-\delta)q + \frac{\delta}{|\cY^{\times n}|} \vec{1}_{\cY^{\times n}} \ ,$$ 
where $\vec{1}$ is the vector of all $1$'s and $q$ is any probability distribution over $\cY^{\times n}$. Given these maps, and under the assumption $c<P<g-s+c$, we have our null hypothesis and alternative hypothesis are:
\begin{align*}
P_0 &= (1-\delta)\ket{s,s,s,s,s,s} + \frac{\delta}{|\cY^{n}|} \vec{1}_{\cY^{n}} \\ P_{1} &= (1-\delta) \left[ (1-q) \ket{s,s,x,s,s,s} + \frac{q}{2} \left( \ket{s,x,s,s,s,s} + \ket{s,s,s,x,s,s} \right) \right] + \frac{\delta}{|\cY^{n}|} \vec{1}_{\cY^{n}} \ ,
\end{align*}
where $x = P - c + s$. Then using the definition the of the KL divergence along with the fact that $P_0(x) \neq P_1(x)$ only for the four specificied sequences, we have
\begin{align*}
    D(P_0||P_1) = \begin{cases} 
                        -\left[1-\delta + \kappa \right]\log(\frac{\kappa}{1-\delta + \kappa}) - \kappa \left[ \log(1 + \frac{(1-\delta)(1-q)}{\kappa}) + 2 \log(1 +  \frac{q(1-\delta)}{2\kappa}) \right] & \text{ if } P \neq c \\
                        0 & \text{ if } P = c 
                  \end{cases} \ ,
\end{align*} 
 where $\kappa := \frac{\delta}{|\cY^{n}|}$. Noting that $c < P$ was assumed, we see that the relative entropy is always a non-zero constant, which completes the derivation.

\subsubsection*{Overhead Meteor Numerics Derivation}
As noted in the main text, a Poisson distribution $p: p(n) = \frac{(\lambda \Delta t)^{n}}{n!}e ^{\lambda \Delta t}~\forall n \in \mathbb{N}$  has a countably infinite number of possible outcomes, but our numerical method requires finite dimensional distributions. This can be handled in the following manner: first we consider an alternate distribution $\tilde{p}$ which is the same as the Poisson distribution until an arbitrarily large but finite number $\tilde{n}-1$. Then $\tilde{p}(\tilde{n}) = 1-F_{p}(\tilde{n}-1)$ where $F_{p}(\cdot)$ is the cumulative distribution of the Poisson. Thus $\tilde{p}$ is a probability distribution on a finite number of mass points and can be made arbitrarily close to the original Poisson distribution. One could then define the alternative hypothesis with regards to $\tilde{p}$. As $\tilde{p}$ is still too large, one chooses a maximum cut-off number $n$ and projects the distributions onto this truncated space without re-normalizing. This projection is a linear trace non-increasing map, and thus for a fixed $\varepsilon \in (0,1)$, the error $\beta_{n}$ can only increase by constructing the decision function on this truncated distribution given Lemma \ref{lemma:DPTNI}. Using this truncated distribution we construct the results for our simple example without loss of rigour.

\bibliography{SETI-Bibliography.bib}

\end{document}